\documentclass[a4paper,reqno,11pt]{amsart}
\usepackage{amsmath, graphicx}
\usepackage{amstext,amsfonts,amsbsy,eucal,amssymb}
\usepackage{color}

\numberwithin{equation}{section}

\newtheorem{theorem}{Theorem}[section]
\newtheorem{lemma}[theorem]{Lemma}
\newtheorem{proposition}[theorem]{Proposition}

\theoremstyle{definition}

\newtheorem{remark}[theorem]{Remark}

\newcommand{\R}{\operatorname{\mathbb{R}}}
\newcommand{\Z}{\operatorname{\mathbb{Z}}}

\newcommand{\Q}{\operatorname{\mathbb{Q}}}
\newcommand{\cS}{\operatorname{\mathcal{S}}}

\newcommand{\bZ}{\operatorname{{\bf Z}}}

\def\<{\operatorname{\langle}}
\def\>{\operatorname{\rangle}}

\def\Ibf{\mathbb I}

\makeatletter
\def\iddots{\mathinner{\mkern1mu\raise\p@
    \hbox{.}\mkern2mu\raise4\p@\hbox{.}\mkern2mu
    \raise7\p@\vbox{\kern7\p@\hbox{.}}\mkern1mu}}
\def\adots{\mathinner{\mkern2mu\raise\p@\hbox{.} %% yhmath.sty'©'ç
 \mkern2mu\raise4\p@\hbox{.}\mkern1mu
 \raise7\p@\vbox{\kern7\p@\hbox{.}}\mkern1mu}}
\makeatother

\begin{document}

\baselineskip 16pt

\title{\large 
Tropical Jacobian and the generic fiber of the ultra-discrete
periodic Toda lattice are isomorphic
}

\author{Rei Inoue}
\address{Faculty of Pharmaceutical Sciences,
Suzuka University of Medical Science, 
\newline\phantom{iii}3500-3 Minami-tamagaki, Suzuka, Mie, 513-8670, Japan}
\email{reiiy@suzuka-u.ac.jp}

\author{Tomoyuki Takenawa} 
\address{Faculty of Marine Technology, 
Tokyo University of Marine Science and Technology,  
\newline\phantom{iii}2-1-6 Etchu-jima, Koto-ku, Tokyo, 135-8533, Japan}
\email{takenawa@kaiyodai.ac.jp }

\begin{abstract}
We prove that the general isolevel set of the ultra-discrete periodic Toda
lattice is isomorphic to the tropical Jacobian associated with
the tropical spectral curve. This result implies that 
the theta function solution
obtained in the authors' previous paper is the complete solution.
We also propose a method to solve the initial value problem.
\end{abstract}

\keywords{tropical geometry, Riemann's theta function, 
Toda lattice}

%\renewcommand{\subjclassname}{%
%  \textup{2000} Mathematics Subject Classification}
%\subjclass{Primary: 37J35. Secondary: 14H70, 14H40.}

\maketitle

%%%%%%%%%%%%%%%%%%%%%%
\section{Introduction}
%%%%%%%%%%%%%%%%%%%%%%

\subsection{Background and aim}

The periodic box and ball system (pBBS) \cite{YuraTokihiro02}
and the ultra-discrete periodic Toda lattice (UD-pToda)
\cite{NagaiTokihiroSatsuma98}
are important examples of 
integrable cellular automata, and they are closely related
\cite{KimijimaTokihiro02}.
The initial value problem of the pBBS was solved in many ways:
by the elementary combinatorics \cite{MIT06},
by the combinatorial Bethe ansatz associated with 
the quantum group \cite{KTT06},
and by the ultradiscretization of the solution to 
the discrete periodic Toda lattice 
\cite{IwaoTokihiro07}.
On the way to solve this problem in \cite{KTT06, KS06,IwaoTokihiro07}
they introduced what now called the tropical Jacobian and 
the tropical Riemann theta function,
and they found that there is a bijection between 
the general isolevel set of the pBBS and the 
lattice points of the tropical Jacobian.

In \cite{InoueTakenawa08} we explored the UD-pToda and the pBBS
from the viewpoint of the tropical geometry.
(See \cite{MikhaZhar06} and references therein for the detail
of the tropical geometry.)    
We conjectured that there is an isomorphism between 
the general isolevel set of the UD-pToda and the tropical Jacobian
of the tropical spectral curve.
The aim in this article is to prove this conjecture (Theorem \ref{J-Toda}),
as a sequel of \cite{InoueTakenawa08,InoueTakenawa-pre}. 
As a consequence, we solve the initial value problem of the UD-pToda.

%%%%%%%%%%%%%%%%%%%%%%%%%%%%%%%%%%%%%%%%%%%
\subsection{Ultra-discrete periodic Toda lattice}

Fix a positive integer $g$.
The ultra-discrete periodic Toda lattice (UD-pToda) is given by
the phase space $\mathcal{T}$: 
$$
\mathcal{T} =
\{\tau = (Q_0,\cdots,Q_{g},W_0,\cdots,W_{g}) 
\in \R^{2(g+1)} ~|~ \sum_{n=0}^{g} Q_n < \sum_{n=0}^{g} W_n \},
$$
and the piecewise-linear map on $\mathcal{T}$:
\begin{align}
T: \mathcal{T} \to \mathcal{T} ; ~(Q_n^t,W_n^t)_{n=0,\cdots,g} \mapsto 
     (Q_n^{t+1},W_n^{t+1})_{n=0,\cdots,g},
\end{align}
where
\begin{align}\label{UD-pToda}
  \begin{split}
  &Q_n^{t+1} = \min[W_n^t, Q_n^t-X_n^t],
  \qquad 
  W_n^{t+1} = Q_{n+1}^t+W_n^t - Q_n^{t+1},
  \\
  &X_n^t = \min_{k=0,\ldots,g}[\sum_{l=1}^k (W_{n-l}^t - Q_{n-l}^t)].
  \end{split}
\end{align}
Here we assume the periodicity $Q_{n+g+1}^t = Q_{n}^t$ and 
$W_{n+g+1}^t = W_{n}^t$.
The map $T$ has conserved quantities given by the tropical polynomials 
$C_i(\tau) ~(i=-1,0,\cdots,g)$ on $\mathcal{T}$, which come from
the coefficients of the characteristic polynomial of 
the Lax matrix through the ultra-discrete limit.
For example, we have
\begin{align}
&C_{-1}(\tau)=\sum_{n=0}^g (Q_n+W_n), \quad 
C_0(\tau)= \min [\sum_{n=0}^g Q_n, \sum_{n=0}^g W_n], \nonumber \\
&C_{g-1}(\tau)=\min[\min_{i\neq j}(Q_i+Q_j) ,
\min_{i \neq j}(W_i+W_j) , 
\min_{j\neq i-1,i}(Q_i+W_j) ]  \label{C's}, \\ 
&C_g(\tau)=\min_{n=0,\ldots,g} [Q_n,W_n]. \nonumber
\end{align}
See \cite{InoueTakenawa08} for a detail of $C_i(\tau)$.
For a fixed $C = (C_{-1},C_0,\cdots,C_g) \in \R^{g+2}$,
we define the isolevel set $\mathcal{T}_C$ as
\begin{align}
  \mathcal{T}_C = \{ \tau \in \mathcal{T} ~|~ C_i(\tau) = C_i 
                     ~(i=-1,0,\cdots,g)\}.
\end{align}

%%%%%%%%%%%%%%%%%%%%

\subsection{Bilinear equation}

Let $\cS_t$ ($t\in \Z$) be a subset of infinite dimensional 
space:
$$
  \cS_t = \{T_n^t \in \R ~|~ n \in \Z \},
$$
where $T_n^t$ has a quasi-periodicity; i.e. $T_n^t$ satisfies 
$T_{n+g+1}^t = T_n^t + c_n^t$, where $c_n^t$ satisfies 
\begin{align}\label{c-period}
   ~c_n^t =an+bt+c  
\end{align}
for some $a,b,c \in \R$.
Fix $L \in \R$ such that $2b-a< (g+1)L$ and define 
a map $\phi$ from $\cS_{t} \times \cS_{t+1}$ 
to $\cS_{t+1} \times \cS_{t+2}$ as
$\phi: (T_n^t,T_n^{t+1})_{n \in \Z} \mapsto
 (T_n^{t+1},T_n^{t+2})_{n \in \Z}$
with 
\begin{align} \label{UD-ptau}
  &T_n^{t+2}=2 T_n^{t+1} -T_n^t + X_{n+1}^t,
\end{align}
where we define a function on $\cS_{t} \times \cS_{t+1}$:
\begin{align}\label{X-tau} 
  X_n^t 
  = \min_{l=0,\cdots,g}\bigl[l L + 2 T_{n-l-1}^{t+1} + T_n^t+T_{n-1}^t
          - (2 T_{n-1}^{t+1} + T_{n-l}^t + T_{n-l-1}^t)\bigr].
\end{align}

\begin{proposition}\cite[Lemma 3.2]{InoueTakenawa-pre}
\\
(i) Let $\sigma_t$ be a map $\sigma_t : 
\cS_{t} \times \cS_{t+1} \to \mathcal{T}$ 
given by
\begin{align}\label{WQ-T}
  W_n^t = L + T_{n-1}^{t+1} + T_{n+1}^{t} - T_{n}^{t} - T_{n}^{t+1}+d,
  \quad
  Q_n^t = T_{n-1}^{t} + T_{n}^{t+1} - T_{n-1}^{t+1} - T_{n}^{t}+d,
\end{align}
where $d \in \R$.
Then the following diagram is commutative:
  \begin{align}\label{sigma-T}
    \begin{matrix}
    \cS_{t} \times \cS_{t+1} & \stackrel{\sigma_t}{\to} &  \mathcal{T} \\[1mm] 
    ~\downarrow_\phi & & \downarrow_{T} \\[1mm]
    \cS_{t+1} \times \cS_{t+2} & \stackrel{\sigma_{t+1}}{\to} & \mathcal{T}
    \end{matrix}~~.
  \end{align}
(ii) The UD-pToda \eqref{UD-pToda} can be transformed into a Hirota type 
bilinear equation for the quasi-periodic function $T_n^t$, 
\begin{align}\label{UD-tau}
  T_{n}^{t-1} +  T_{n}^{t+1}
  = 
  \min[2 \, T_{n}^{t}, ~T_{n-1}^{t+1} + T_{n+1}^{t-1} +L].
\end{align}
\end{proposition}

%%%%%%%%%%%%%%%%%%%%%%%%%%%%%

\subsection{Main result}
 
We assume a generic condition for $C=(C_{-1},C_0,\cdots,C_g) \in \R^{g+2}$: 
\begin{align}
  \label{CD-condition}  
  C_{-1} > 2 C_0, ~  
  C_i+C_{i+2} > 2 C_{i+1} ~ (i=0,\cdots,g-2),~
  C_{g-1} > 2 C_g.
\end{align}
Define $L$, $\lambda_0,\lambda_1,\cdots,\lambda_g$ and 
$p_1,\cdots,p_g$ by 
\begin{align}
  \label{partition}
&  L = C_{-1}-2(g+1)C_g, \qquad \lambda_0=C_g, \qquad  
  \lambda_i = C_{g-i}-C_{g-i+1} ~~ (1 \leq i \leq g), \\
&  p_0 = L, \quad
   p_i = L - 2 \sum_{j=1}^g \min[\lambda_i-\lambda_0,\lambda_j-\lambda_0]
  ~~ (1 \leq i \leq g).
\end{align}

Fix $C \in \R^{g+2}$ which satisfies the generic condition 
\eqref{CD-condition}.
Then we have 
$$
  \lambda_0 < \lambda_1 < \lambda_2 < \cdots < \lambda_g, 
  \quad 
  p_0 > p_1 > \cdots > p_g > 0.
$$
Define a positive definite symmetric matrix $K \in M_g(\R)$ as
\begin{align}
K_{ij} &=
\left\{\begin{array}{lcl}
p_{i-1}+p_i+2(\lambda_i-\lambda_{i-1}) &\ &(i=j)\\
-p_i && (j=i+1)\\
-p_j && (i=j+1)\\
0 &&(\mbox{otherwise}).
\end{array}\right.
\end{align}
In fact, $K$ is the period matrix for the tropical curve $\Gamma_C$
determined by $C$. 
See \cite{InoueTakenawa08} for the definition of $\Gamma_C$.
In parallel with the complex algebraic geometry,
we have the tropical Jacobian $J(K)=\R^g/(\Z^g K)$ 
and the tropical Riemann theta function 
$\Theta(\bZ) = \Theta(\bZ;K)$ on $\R^g$: 
$$
\Theta(\bZ) 
= \min_{{\bf m}\in \Z^g} 
  [ \frac12 {\bf m} K{\bf m}^{\bot}+{\bf m}{\bf Z}^{\bot}], \qquad 
({\bf Z}\in \R^g).$$

\begin{proposition}\label{bilinear-theta}
\cite[Corollary 2.13, Theorem 3.5]{InoueTakenawa-pre}
\\
(i) Set $\vec{\lambda} =(\lambda_1-\lambda_0, \lambda_2-\lambda_1,
\cdots,\lambda_g-\lambda_{g-1})$ and  ${\bf e}_1 = (1,0,\cdots,0)$.
For ${\bf Z}_0 \in \R^g$, the function $T_n^t$ given by
\begin{align}
  \label{tau-theta}
  &T_n^t =   \Theta({\bf Z}_0- n L {\bf e}_1+t\vec{\lambda}),
\end{align}
satisfies \eqref{UD-tau}. 
\\
(ii) 
Let $\iota_t$ be the map $\R^g \to \cS_t \times \cS_{t+1}$ defined by
  $$ 
    {\bf Z}_0 \mapsto 
    (T_n^t= \Theta({\bf Z}_0- n L {\bf e}_1+t\vec{\lambda}),~
     T_n^{t+1}=\Theta({\bf Z}_0- n L {\bf e}_1+(t+1)\vec{\lambda}))_{n \in \Z},
  $$
then the following diagram is commutative:
 \begin{align}\label{iota-T}
    \begin{matrix}
    \R^g & \stackrel{\iota_t}{\to} & \cS_{t} \times \cS_{t+1}
    & \stackrel{\sigma_t}{\to} &  \mathcal{T}_C 
    \\[1mm] 
    ~\downarrow_{{\rm id.}} & & ~\downarrow_{\phi}
    & & \downarrow_{T} \\[1mm]
    \R^g & \stackrel{\iota_{t+1}}{\to} & \cS_{t+1} \times \cS_{t+2} & 
    \stackrel{\sigma_{t+1}}{\to} & \mathcal{T}_C
    \end{matrix},
  \end{align} 
where we have $d= C_g$ for $\sigma_t$ \eqref{WQ-T}. 
(See Remark \ref{rem:period}.)
Namely, \eqref{tau-theta} 
gives a solution to \eqref{UD-pToda} through \eqref{WQ-T}.
\end{proposition}

Our goal of this article is
to show the following theorem
\begin{theorem}\label{J-Toda}
  Let $\iota_\sigma : J(K) \to \mathcal{T}_C$ 
  be the map induced by $\sigma_t \circ \iota_t : \R^g \to \mathcal{T}_C$. 
  Then the map $\iota_\sigma$ is isomorphic.
\end{theorem}

We will prove Theorem \ref{J-Toda} in the following order:
we show $\sigma_t |_{{\rm Im \iota_t}}$ is injective 
(Proposition~\ref{sigma_inj}) in \S \ref{part1}, 
and $\iota_t$ is injective (Proposition~\ref{iota_inj}) in \S \ref{part2}.
Finally we prove Theorem \ref{J-Toda} in \S \ref{part3}.
To help to understand the proof of the theorem
(especially for \S \ref{part2}),
we show some idea and an example in \S \ref{example}. 
In \S \ref{init-prob} we summarize the initial value problem
for UD-pToda.

%%%%%%%%%%%%%%%%%%%%%%%%%%%%%%%%%%
\subsection{Notations and Remarks}
%%%%%%%%%%%%%%%%%%%%%%%%%%%%%%%%%%

We use the following notations of vectors in $\R^g$:
\begin{align*}
\vec{\lambda} &=(\lambda_1-\lambda_0, \lambda_2-\lambda_1,
\cdots,\lambda_g-\lambda_{g-1}), \\
\vec{g}&=(g,g-1,\dots,1),\\
{\bf e}_i&: \mbox{the $i$-th vector of standard basis of } \R^g,\\
\Ibf&=(1,1,\dots,1)={\bf e}_1+\cdots+{\bf e}_g,\\
\Ibf_k&=(1,\dots,1,0,\dots,0)={\bf e}_1+\cdots+{\bf e}_k.
\end{align*}

\begin{remark}
For $C \in \R^{g+2}$ with the condition \eqref{CD-condition}, 
we have the following relations:
\begin{align}
  L> 2 \sum_{i=1}^g (\lambda_i -\lambda_0),
  \qquad 
  \vec{g} \cdot \vec{\lambda}^{\bot}
  =\sum_{i=1}^g (\lambda_i -\lambda_0)=C_0-(g+1)C_g \label{g-lambda}
\end{align}
and 
\begin{align*}
  \sum_{j=1}^g K_{ij} > 0, \qquad
  \mathbb{I} K=
%&\left(L+2(\lambda_1-\lambda_0), 2(\lambda_2-\lambda_1),\dots,
%2(\lambda_{g-1}-\lambda_{g-2}), 2(\lambda_g-\lambda_{g-1})+p_g\right)^{\bot}\\
  p_g {\bf e}_g+ L {\bf e}_1+ 2\vec{\lambda}, \qquad
  \vec{g}K=(g+1)L{\bf e}_1.
\end{align*}
\end{remark}

\begin{remark}\label{rem:period}
The tropical Riemann theta function $\Theta({\bf Z})$ is quasi-periodic:
$$\Theta({\bf Z}+{\bf l}K)=
-\frac12 {\bf l}K{\bf l}^{\bot}-{\bf l}{\bf Z}^{\bot}+\Theta({\bf Z})  
\quad ({\bf l}\in \Z^g)$$
and even: $\Theta(-{\bf Z})=\Theta({\bf Z}).$
Thus, if we set $T_n^t$ as \eqref{tau-theta},
it satisfies the following quasi-periodicity;
  \begin{align}\label{T-qperiod}
  T_{n+g+1}^t = T_n^t + c_n^t, 
  \qquad c_n^t =an+bt+c= \vec{g} \cdot 
  (\bZ_0-nL\vec{{\bf e}_1}+\vec{\lambda} t - \frac{1}{2} \vec{g} K)^\bot,
  \end{align}
which yields
$$
  a=-gL, \quad b=\sum_{i=1}^g (\lambda_i-\lambda_0).
$$
Then we obtain $d=C_g$ at \eqref{WQ-T}, via \eqref{partition} and 
$$C_{-1}=\sum_{n=0}^g(Q_n^t+W_n^t)=(g+1)L+2(g+1)d+a.$$
\end{remark}

%%%%%%%%%%%%%%%%%%%%%%%%%%%%%%%%%%%%%%%%%%%%%%%%%%%%%%%%%%%
\section{Proof of Theorem \ref{J-Toda}}
%%%%%%%%%%%%%%%%%%%%%%%%%%%%%%%%%%%%%%%%%%%%%%%%%%%%%%%%%%%

In this section, we assume $C_g =0$
without loss of generality.
For $\mathbf{m}=(m_1,m_1,\ldots,m_g) \in \Z^g$,
set the fundamental region $D_{\bf m}$ of $\Theta({\bf Z})$ as
$D_{\bf m}=
\{{\bf Z}\in \R^g \ ;\  \Theta({\bf Z})=\frac12 {\bf m}K{\bf m}^{\bot}
+{\bf m}{\bf Z}^{\bot} \}$
which is explicitly written as 
\begin{align}\label{region}
  D_{\bf m}=&\{{\bf Z}\in \R^g ~|~ -{\bf l} {\bf Z}^{\bot} 
\leq  {\bf l}K({\bf m}+\frac12 {\bf l})^{\bot}\\
&\quad  \mbox{ for any ${\bf l}=\pm({\bf e}_j+{\bf e}_{j+1}+\cdots+{\bf e}_k)$ 
 and $1\leq j\leq k \leq g$} \} \nonumber
\end{align} 
from Lemma~2.9 of \cite{InoueTakenawa-pre}.
For a fixed ${\bf Z}_0 \in \R^g$, we write 
${\bf Z}_n^t={\bf Z}_0-nL{\bf e}_1+t\vec{\lambda}$. 
Due to the structure of the fundamental regions,
it is easy to see that if ${\bf Z}_n^t \in D_{{\bf m}}$, then 
\begin{align}\label{period-D}
{\bf Z}_{n+g+1}^t \in D_{{\bf m}+\vec{g}},
\qquad 
\bZ_n^t = \bZ_{n+g+1}^t \text{ mod } \Z^g K.  
\end{align}

%%%%%%%%%%%%%%%%%%%%
\subsection{Idea and example}
%%%%%%%%%%%%%%%%%%%%
\label{example}

We first recall the relation between the UD-pToda and the pBBS.
The pBBS is a cellular automaton
that the finite number of balls move in a periodic array of $L$ boxes
each of which has one ball at most \cite{YuraTokihiro02}.
This system has conserved quantities parameterized by
a non-decreasing array
$\lambda = (\lambda_1,\cdots,\lambda_g) \in (\Z_{> 0})^g$.
We write $0$ and $1$ for 
``an empty box" and ``an occupied box" respectively.
Let $B_L \simeq \{0,1\}^{\times L}$ 
be the phase space of $L$-periodic BBS and
$B_{L,\lambda} \subset B_L$ be a set of the states 
whose conserved quantity is $\lambda$. 
Then the injection 
$\beta: ~B_{L,\lambda} \hookrightarrow \mathcal{T}_C \cap \Z^{2g+2};
~ b \mapsto (Q_0,W_0,Q_1,W_1,\cdots,Q_{g},W_{g})$ is defined as follows:
\begin{enumerate} 
\item 
if the leftmost entry of $b$ is $1$, then set 
$Q_0 = \sharp(\text{the first consecutive $1$'s from the left})$,
otherwise set $Q_0 = 0$.

\item
Set $W_i = \sharp(\text{the $(i-1)$-th consecutive $0$'s from the left})$
for $i=0,\cdots, g$.
If $Q_0 > 0$,
set $Q_i = \sharp(\text{the $(i-1)$-th consecutive $1$'s from the left})$,
otherwise set 
$Q_i = \sharp(\text{the $i$-th consecutive $1$'s from the left})$
for $i=1,\cdots, g$.
\end{enumerate}
The following shows how $\beta$ works;
\begin{align*}
&    \underbrace{1 \cdots 1}_{Q_0} \underbrace{0 \cdots 0}_{W_0}
    ~~\cdots~~ 
    \underbrace{0 \cdots 0}_{W_{g-1}} \underbrace{1 \cdots 1}_{Q_{g}}
    \quad (W_g = 0)\\
&    \underbrace{0 \cdots 0}_{W_0} \underbrace{1 \cdots 1}_{Q_1}
    ~~\cdots~~ 
    \underbrace{1 \cdots 1}_{Q_{g}} \underbrace{0 \cdots 0}_{W_{g}}
    \quad (Q_0 = 0).
\end{align*}
The map $\beta$ is not surjective but induces an isomorphism
from $B_{L,\lambda}$ to $\mathcal{T}_C \cap \Z^{2g+2}$ 
modulo the shift $(Q_0,W_0,Q_1,W_1,\cdots,Q_{g},W_{g})
\mapsto (Q_1,W_1,\cdots,Q_{g},W_{g},Q_0,W_0)$.
See Section~4 of \cite{InoueTakenawa08} for the details. 

Let us consider the UD-pToda of $g=2$ with the initial values
\begin{align}\label{example-init}
  (Q_0,W_0,Q_1,W_1,Q_2,W_2)=(0,1,2,7,4,0)
\end{align}
at $t=0$. Thus we have
$$
  (X_0^0,X_1^0,X_2^0)=(-4,-3,0), \quad L = 14, \quad \vec{\lambda}=(1,4), 
  \quad K = \begin{pmatrix} 26 & -10 \\ -10 & 20 \end{pmatrix}.
$$
We assume that $T_n^t$ is written as 
$T_n^t = \Theta({\bf Z}_n^t = {\bf Z}_0 - n L {\bf e}_1 + \vec{\lambda} t)$
for all $n,t \in \Z$.
From $X_2^0=0$ and \eqref{UD-ptau}, 
the points 
$({\bf Z}_1^0,\Theta({\bf Z}_1^0)), ({\bf Z}_1^1,\Theta({\bf Z}_1^1)), 
({\bf Z}_1^2,\Theta({\bf Z}_1^2))$ are colinear and thus
the points ${\bf Z}_1^0, {\bf Z}_1^1, {\bf Z}_1^2$ 
belong to the same region $D_{\bf m}$.
Thus the points ${\bf Z}_{-2}^0, {\bf Z}_{-2}^1, {\bf Z}_{-2}^2$ 
belong to the same region $D_{{\bf m} - (2,1)}$ from \eqref{period-D}.
By assuming ${\bf m} = (2,1)$, in this case 
we can fortunately find ${\bf Z}_0=(-28,-3)$ by a heuristic way.
See Fig.~\ref{ex1a}.
The following is the time evolution of \eqref{example-init}: 
$$
\begin{array}{ccc}
t & b^t &(Q_0^t,W_0^t,Q_1^t,W_1^t,Q_2^t,W_2^t)\\[1mm]
0&{}_{s}01100000001111&(0,1,2,7,4,0)\\
1&{}_{s}10011111000000&(1,2,5,6,0,0)\\
2&0_{s}1000000111110&(1,6,5,1,0,1)
\end{array}
$$
where ${}_s$ denotes the place of $Q_0$ in the corresponding state 
of the pBBS.

\begin{figure}
\begin{center}
\unitlength=1mm
\begin{picture}(100,80)(-50,-30)
\put(0,0){\vector(1,0){40}}
\put(0,0){\line(-1,0){60}}
\put(0,0){\vector(0,1){38}}
\put(0,0){\line(0,-1){30}}
\put(40,3){$Z_1$}
\put(3,40){$Z_2$}

\thicklines

%D_{(0,0)}
\put(13,0){\line(-1,1){10}}
\put(3,10){\line(-1,0){16}}
\put(-13,10){\line(0,-1){10}}
\put(-13,0){\line(1,-1){10}}
\put(-3,-10){\line(1,0){16}}
\put(13,-10){\line(0,1){10}}
\put(10,-8){$ \longleftarrow D_{(0,0)}$}

%D_{(1,0)}
\put(-13,10){\line(-1,1){10}}
\put(-23,20){\line(-1,0){16}}
\put(-39,20){\line(0,-1){10}}
\put(-39,10){\line(1,-1){10}}
\put(-29,0){\line(1,0){16}}
%\put(-13,0){\line(0,1){10}}
\put(-15,-15){$D_{(1,1)}$}

%D_{(1,1)}
%\put(13,0){\line(-1,1){10}}
%\put(3,10){\line(-1,0){16}}
\put(-29,0){\line(0,-1){10}}
\put(-29,-10){\line(1,-1){10}}
\put(-19,-20){\line(1,0){16}}
\put(-3,-20){\line(0,1){10}}
\put(-33,13){$D_{(1,0)}$}

%D_{(2,1)}
%\put(-29,0){\line(-1,1){10}}
\put(-39,10){\line(-1,0){16}}
\put(-55,10){\line(0,-1){10}}
\put(-55,0){\line(1,-1){10}}
\put(-45,-10){\line(1,0){16}}
%\put(13,-10){\line(0,1){10}}
\put(-53,5){$D_{(2,1)}$}

\thinlines

{\color{red}
\put(0,-3){\line(-1,0){42}}
\put(1,1){\line(-1,0){42}}
\put(2,5){\line(-1,0){42}}
\put(0,-3){\line(1,4){2}}
\put(-14,-3){\line(1,4){2}}
\put(-28,-3){\line(1,4){2}}
\put(-42,-3){\line(1,4){2}}
{\scriptsize
\put(0,-7){${\bf Z}_{-2}^0$}
\put(-14,-7){${\bf Z}_{-1}^0$}
\put(-28,-7){${\bf Z}_{0}^0$}
\put(-42,-7){${\bf Z}_{1}^0$}
\put(2,-2){${\bf Z}_{-2}^1$}
\put(3,4){${\bf Z}_{-2}^2$}
}
}
\end{picture}
\caption{$D_{{\bf m}}$ and ${\bf Z}_n^t$}\label{ex1a}
\end{center}
\end{figure}

\noindent {\bf Observations.}\
From Fig.~\ref{ex1a} we see the following facts:
\\
$\cdot$ 
${\bf Z}_n^t \in D_{{\bf m}}$ implies ${\bf Z}_{n+1}^t \in D_{{\bf m}}$,
$D_{{\bf m}+(1,0)}$ or $D_{{\bf m}+(1,1)}$ (Lemma~\ref{regionZ}),
\\
$\cdot$ 
${\bf Z}_{-1}^0, {\bf Z}_{0}^0 \in D_{(1,1)}$
while $Q_0^0=W_{2}^0=0$ (Lemma~\ref{QW0}).

From these observations we see the following:
even if we do not know ${\bf Z}_0$, 
once the quantities of $T_{n}^0$ and 
the domain $D_{{\bf m}_n}$ to which ${\bf Z}_n^0$ belongs are given
for all $n\in \Z/3\Z$,
we have two independent linear equations for ${\bf Z}_0$.
Actually, from
${\bf Z}_{-1}^0 \in D_{(1,1)}$ and ${\bf Z}_{1}^0 \in D_{(2,1)}$
we have 
\begin{align*}
\frac{1}{2}(1,1)K(1,1)^{\bot} +(1,1)({\bf Z}_0+(L,0))^{\bot}&=T_{-1}^0,\\
\frac{1}{2}(2,1)K(2,1)^{\bot} +(2,1)({\bf Z}_0-(L,0))^{\bot}&=T_1^0,\\
\end{align*}
and ${\bf Z}_0$ is uniquely determined.
 
Therefore our problem is how to determine the region $D_{{\bf m}_n}$
to which ${\bf Z}_n^t$ belongs when $T_n^t$ are given.
For this purpose
we investigate what happens
when we replace $W_{n_0}^0$ with $\tilde{W}_{n_0}^0=W_{n_0}^0+\delta$ 
by using $\delta \gg 1$,
for $n_0$ as $X_{n_0+1}^0 = 0$.

In the previous example we have $n_0 = -2$ and set $\delta=10$:
$$
\begin{array}{ccc}
t & b^t & 
(\tilde{Q}_0^t,\tilde{W}_0^t,\tilde{Q}_1^t,\tilde{W}_1^t,\tilde{Q}_2^t,
\tilde{W}_2^t)
\\[1mm]
-3 & 0001111100000000000000{}_{s}10 & (1,1,0,3,5,14)\\
-2 & 00000000011111000000000{}_{s}1 & (1,0,0,9,5,9)\\
-1 & 100000000000000111110000{}_{s} & (0,0,1,14,5,4)\\
0  &{}_{s}011000000000000000001111 & (0,1,2,17,4,0)\\
1 &{}_{s}100111110000000000000000 & (1,2,5,16,0,0)\\
2 &0{}_{s}10000001111100000000000 & (1,6,5,11,0,1)\\
3 &00{}_{s}1000000000011111000000 & (1,11,5,6,0,2)\\
\end{array}
$$

\begin{figure}
\begin{center}
\unitlength=1mm

\begin{picture}(120,100)(-80,-40)
\put(0,0){\vector(1,0){40}}
\put(0,0){\line(-1,0){110}}
\put(0,0){\vector(0,1){48}}
\put(0,0){\line(0,-1){45}}
\put(40,3){$\tilde{Z}_1$}
\put(3,50){$\tilde{Z}_2$}

%D_{(0,0)}
\put(23,0){\line(-1,1){20}}
\put(3,20){\line(-1,0){26}}
\put(-23,20){\line(0,-1){20}}
\put(-23,0){\line(1,-1){20}}
\put(-3,-20){\line(1,0){26}}
\put(23,-20){\line(0,1){20}}
\put(10,-10){$\tilde{D}_{(0,0)}$}

%D_{(1,0)}
\put(-23,20){\line(-1,1){20}}
\put(-43,40){\line(-1,0){26}}
\put(-69,40){\line(0,-1){20}}
\put(-69,20){\line(1,-1){20}}
\put(-49,0){\line(1,0){26}}
%\put(23,-20){\line(0,1){20}}
\put(-60,30){$\tilde{D}_{(1,0)}$}

%D_{(1,1)}
%\put(23,0){\line(-1,1){20}}
%\put(3,20){\line(-1,0){26}}
\put(-49,0){\line(0,-1){20}}
\put(-49,-20){\line(1,-1){20}}
\put(-29,-40){\line(1,0){26}}
\put(-3,-40){\line(0,1){20}}
\put(-20,-30){$\tilde{D}_{(1,1)}$}

%D_{(2,1)}
%\put(23,0){\line(-1,1){20}}
\put(-69,20){\line(-1,0){26}}
\put(-95,20){\line(0,-1){20}}
\put(-95,0){\line(1,-1){20}}
\put(-75,-20){\line(1,0){26}}
%\put(23,-20){\line(0,1){20}}
\put(-90,10){$\tilde{D}_{(2,2)}$}

\thinlines
\put(0,-3){\color{red}\line(-1,0){72}}
\put(1,1){\color{red}\line(-1,0){72}}
\put(2,5){\color{red}\line(-1,0){72}}
\put(0,-3){\color{red}\line(1,4){2}}
\put(-24,-3){\color{red}\line(1,4){2}}
\put(-48,-3){\color{red}\line(1,4){2}}
\put(-72,-3){\color{red}\line(1,4){2}}

\put(-3,-15){\color{blue}\line(-1,0){72}}
\put(-2,-11){\color{blue}\line(-1,0){72}}
\put(-1,-7){\color{blue}\line(-1,0){72}}
\put(3,9){\color{blue}\line(-1,0){72}}
\put(4,13){\color{blue}\line(-1,0){72}}

\put(-3,-15){\color{blue}\line(1,4){3}}
\put(-27,-15){\color{blue}\line(1,4){3}}
\put(-51,-15){\color{blue}\line(1,4){3}}
\put(-75,-15){\color{blue}\line(1,4){3}}
\put(2,5){\color{blue}\line(1,4){2}}
\put(-22,5){\color{blue}\line(1,4){2}}
\put(-46,5){\color{blue}\line(1,4){2}}
\put(-70,5){\color{blue}\line(1,4){2}}

{\scriptsize
\put(-78,-4){\color{red}$\tilde{{\bf Z}}_{1}^0$}
\put(-77,0){\color{red}$\tilde{{\bf Z}}_{1}^1$}
\put(-76,4){\color{red}$\tilde{{\bf Z}}_{1}^2$}

\put(-2,-16){\color{blue}$\tilde{{\bf Z}}_{-2}^{-3}$}
\put(-1,-12){\color{blue}$\tilde{{\bf Z}}_{-2}^{-2}$}
\put(0,-8){\color{blue}$\tilde{{\bf Z}}_{-2}^{-1}$}
\put(1,-4){\color{red}$\tilde{{\bf Z}}_{-2}^0$}
\put(2,0){\color{red}$\tilde{{\bf Z}}_{-2}^1$}
\put(3,4){\color{red}$\tilde{{\bf Z}}_{-2}^2$}
\put(4,8){\color{blue}$\tilde{{\bf Z}}_{-2}^3$}
\put(5,12){\color{blue}$\tilde{{\bf Z}}_{-2}^4$}
}
\end{picture}

\caption{$\tilde{D}_{\vec{m}}$ and $\tilde{{\bf Z}}_n^t$}\label{ex1b}
\end{center}
\end{figure}
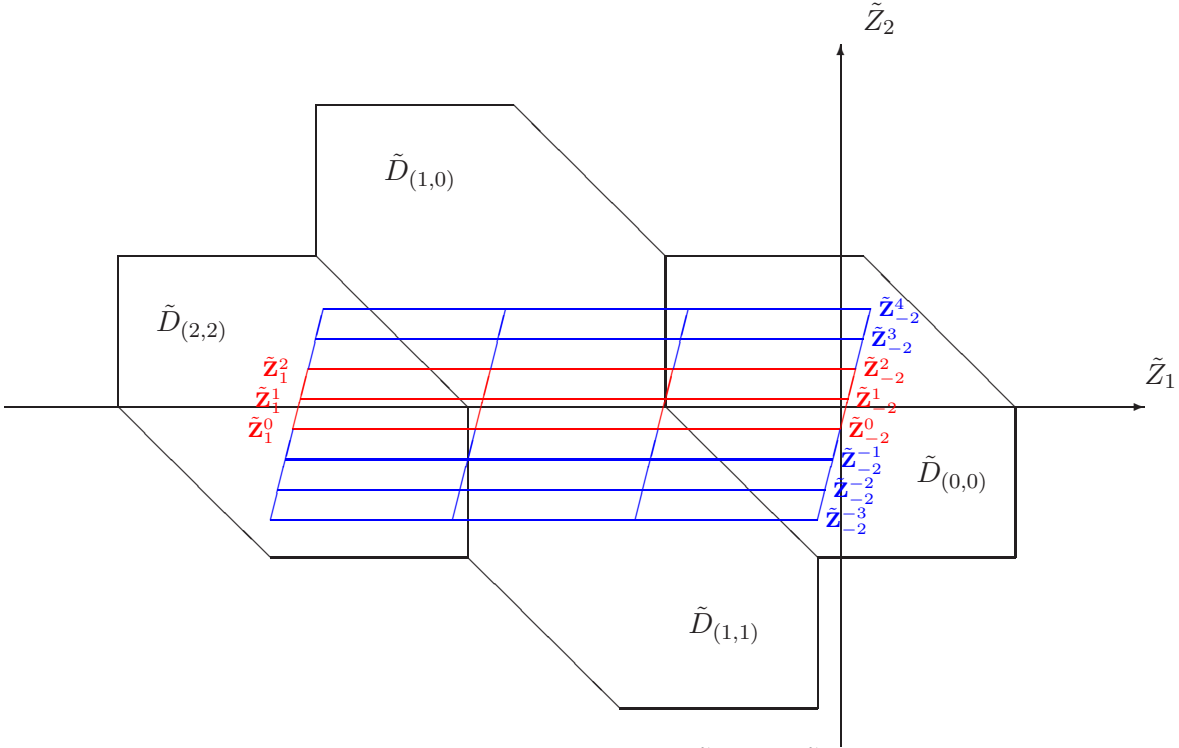

\noindent {\bf Observations.}\\
$\cdot$
By assuming $\tilde{{\bf Z}}_{-2}^0 \in \tilde{D}_{(0,0)}$,
we can find $\tilde{{\bf Z}}_0=(-48,-3)$, thus we have 
${\bf Z}_{-2}^0=\tilde{{\bf Z}}_{-2}^0$.
\\
$\cdot$
For $t=1,2,3$ we have ${\bf m}_{-2} = (0,0)$ (by the assumption),
${\bf m}_{-1}=(0,0)$, ${\bf m}_0 = (1,0)$ and ${\bf m}_1 = (2,1)$.

The second observation implies
$\tilde{{\bf Z}}_{0}^t \in D_{(1,0)}$ and $\tilde{{\bf Z}}_{1}^t \in D_{(2,1)}$.Then $\tilde{{\bf Z}}_0$ is uniquely determined in the same way as the above,
once $\{T_n^t\}_n$ is given.  
By using the first observation $\bZ_0$ is uniquely fixed by $\tilde{\bZ}_0$.

\begin{remark}
This idea of replacing $W_{n_0}^0$ with a huge number
$W_{n_0}^0 + \delta$
(where $X_{n_0+1} = 0$)
would be understandable by considering the 
corresponding pBBS.
In the pBBS the equation $X_{n_0+1}^t=0$ means that
the ``load of carrier'' \cite{TakahashiMatsukidaira97} 
is zero at the beginning ``$1$'' of
$Q_{n_0+1}^t$, and hence $\tilde{Q}_n^{t+1}$ and $\tilde{W}_n^{t+1}$ does not
change from the original ones except $\tilde{W}_{n_0}^{t+1}$.   
Moreover, for a large $t$ 
the clusters of 1's in pBBS will align
coherently in the order of their sizes 
$\lambda_1, \lambda_2, \dots, \lambda_g$ thanks to $\delta$.
This indicates 
$\tilde{Q}_n^t =
\lambda_{n-n_0-1}$ for $n=n_0+1,\dots,n_0+g+1$
and $\tilde{X}_n^t = 0$ for all $n$ in the UD-pToda.
On the other hand from $\tilde{X}_n^t = 0$ and \eqref{WQ-T} we obtain
$\tilde{Q}_n^t = ( {\bf m}_n - {\bf m}_{n-1} ) \lambda^\bot$.
By making use of these properties the region $D_{{\bf m}_n}$
to which $\tilde{\bZ}_n^t$ belongs is obtained for a large $t$ 
as Lemma~\ref{Dtilde}.  
\end{remark}

%%%%%%%%%%%%%%%%%%%%%%%%%%%%%%%%%%%%%%%%%%%%%%%%%
\subsection{Proof of Proposition \ref{sigma_inj}}
%%%%%%%%%%%%%%%%%%%%%%%%%%%%%%%%%%%%%%%%%%%%%%%%%
\label{part1}
Consider $(Q_n^t$, $W_n^t)_n \in \mathcal{T}_C$ 
in the image of $\sigma_t \circ \iota_t$,
i.e. $Q_n^t$ and $W_n^t$ are written in terms of the tropical Riemann 
theta function via \eqref{WQ-T} 
and $T_n^t=\Theta({\bf Z}_n^t ={\bf Z}_0- n L{\bf e}_1+ t\vec{\lambda})$
with some ${\bf Z}_0 \in \R^g$.

\begin{lemma} \label{regionZ} 
If ${\bf Z} \in D_{\bf m}$, 
then ${\bf Z}- n L{\bf e}_1+ t\vec{\lambda}$ ($n=0,-1,\ t=0,1,2$) 
belongs to $D_{{\bf m}'}$, 
where $m_1'-m_1 \in \{ 0,-1\}$.  
\end{lemma}

\begin{proof}
Without loss of generality we can assume ${\bf Z} \in D_{\bf 0}$,
namely $-\frac12{\bf l}K{\bf l}^{\bot} \leq -{\bf l}\bZ^{\bot}
\leq \frac12 {\bf l}K{\bf l}^{\bot}$ for all 
${\bf l}\in \Z^g$.
Define 
$$
  f({\bf l}, {\bf m}') 
  = 
  -{\bf l} ({\bf Z}-nL{\bf e}_1+t\vec{\lambda})^{\bot} 
  -{\bf l}K(\frac12 {\bf l}+{\bf m}')^{\bot}. 
$$
We consider the following four cases; 
(a-1): $m_1'=-2$, (a-2): $m_1'<-2$, (b-1): $m_1'=1$ or (b-2): $m_1'>1$.
We show that in case (a-1) or (b-1) there exists 
${\bf l}\in \Z^g$ such that $f({\bf l}, {\bf m}') \geq 0$
(i.e. ${\bf Z}-nL{\bf e}_1+t \vec{\lambda}$ is not in the interior of 
$D_{{\bf m}'}$), and that in case (a-2) or (b-2) there exists
${\bf l}\in \Z^g$ such that $f({\bf l}, {\bf m}') > 0$ 
(i.e. ${\bf Z}-nL{\bf e}_1+t \vec{\lambda}$ is not in $D_{{\bf m}'}$),
for $n=0,-1$ and $t=0,1,2$.\\
(a-1) or (a-2): Take ${\bf l}$ as $l_i=1$ if $m_i'\leq -2$ and 
$l_i=0$ if $m_i'\geq -1$. Then we get
\begin{align*}
f({\bf l}, {\bf m}') &\geq  
-{\bf l}K{\bf l}^{\bot} +nL-t{\bf l}\vec{\lambda}^{\bot} 
-{\bf l}K{\bf m}'^{\bot} 
\qquad (\text{since } 
-{\bf l}{\bf Z}^{\bot}\geq -\frac12 {\bf l}K{\bf l}^{\bot} )\\
&= nL- t{\bf l}\vec{\lambda}^{\bot}
+{\bf l}K(\Ibf- ({\bf m}'+{\bf l}+\Ibf))^{\bot},
\end{align*} 
where we have $-{\bf l}K({\bf m}'+{\bf l}+\Ibf)^{\bot} \geq 0$ since 
$-(m_i'+l_i+1){\bf l}K{\bf e}_i^{\bot} \geq 0$ for all $i$. Actually,
if $m_i'\leq -2$, then $-(m_i'+l_i+1) \geq 0$ and
${\bf l} K {\bf e}_i^{\bot}= K_{ii}+\sum_{j\neq i} l_j K_{ij}> 0$, 
and if $m_i'\geq -1$, then $-(m_i'+l_i+1) \leq 0$ and
${\bf l} K {\bf e}_i^{\bot}= \sum_{j\neq i} l_j K_{ij}\leq 0$.
Using ${\bf l}K\Ibf^{\bot} 
= {\bf l} ((L,0,\ldots,0,p_g) + 2 \vec{\lambda})$, 
we obtain
\begin{align*}
nL- t{\bf l}\vec{\lambda}^{\bot}
+{\bf l}K(\Ibf- ({\bf m}'+{\bf l}+\Ibf))^{\bot}
&\geq nL+ {\bf l}((L,0,\dots,0,p_g) + (2-t)\vec{\lambda})^{\bot}\\
&= (1+n)L+ l_gp_g +(2-t){\bf l}\vec{\lambda}^{\bot} 
   \qquad ( \text{since } l_1 = 1)
\end{align*}
which is non-negative 
for $n=0,-1$ and $t=0,1,2$. Thus the claim follows.
In the case of (a-2), $f({\bf l}, {\bf m}') > 0$ can be shown similarly.\\
(b-1) or (b-2): Take ${\bf l}$ as $l_i=-1$ if $m_i'\geq 1$ and 
$l_i=0$ if $m_i'\leq 0$. In the same way as (a-1) we can show 
$f({\bf l}, {\bf m}') \geq 0$ and $f({\bf l}, {\bf m}') > 0$ respectively.
\end{proof}

\begin{lemma}\label{Xn0}
(i) For any $t \in \Z$, 
there exists $n_0 \in \Z/(g+1)\Z$ such that $X_{n_0+1}^t=0$.
\\
(ii) Fix ${\bf Z}_0\in \R^g$ and set ${\bf Z}_n^t 
= {\bf Z}_0-nL{\bf e}_1 + \vec{\lambda} t$.
If ${\bf Z}_0 \in D_{\bf m}$, there exists an unique $n_1 \in \{1,\cdots,g+1\}$
such that ${\bf Z}_{n_1}^0,{\bf Z}_{n_1-1}^0\in D_{{\bf m}'}$ with $m_1'-m_1=n_1-1$. 
\end{lemma}
\begin{proof}
(i) Set the sequence $\{a_k\}_{k=1,2,3,\dots}$ as  
$a_k= \sum_{l=1}^k  (W_{-l}^t-Q_{-l}^t)$
and find the term which is smaller than all rear terms.
From the assumption $\sum_n Q_n^t <\sum_n W_n^t$, such term
exists within the first $g+1$ terms. Let $a_{n_0}$ be a such term
and set the sequence $\{b_k\}_{k=1,2,3,\dots}$ as
$b_k= a_{n_0+k}-a_{n_0}=\sum_{l=1}^k  (W_{-n_0-l}^t-Q_{-n_0-l}^t).$
Then $b_k \geq 0$ holds for all $k$, thus we have
$X_{-n_0}^t= \min [0, b_1,b_2,\dots,b_g ]=0$.\\
(ii)
When ${\bf Z}_0^0 = {\bf Z}_0 \in D_{\bf m}$, we have 
${\bf Z}_{g+1}^0\in D_{{\bf m}+\vec{g}}$.
From Lemma~\ref{regionZ},
each point ${\bf Z}_n^0 ~(n=1,\dots,g+1)$
is in $D_{{\bf m}_n}$ with $({\bf m}_n)_1- ({\bf m}_{n-1})_1 \in \{0,1\}$.
Thus the claim follows from the pigeonhole principle.
\end{proof}

\begin{lemma}\label{QW0}
For $(Q_n^0,W_n^0)_n \in \mathcal{T}_C$, 
there uniquely exists $n_1 \in \Z/(g+1)\Z$ 
such that either i) or ii) in the following is satisfied:
$$
  \text{i)}  ~ Q_{n_1}^0=0, 
  \qquad
  \text{ii)} ~ W_{n_1-1}^0=0, ~Q_{n_1-1}^0>0 \text{ and } Q_{n_1}^0>0.
$$
Further, let ${\bf Z}_0$ be a point in $\R^g$ such that 
$\sigma_0 \circ \iota_0 ({\bf Z}_0) = (Q_n^0,W_n^0)_n$.
Then we have the following properties: 
$n_1$ coincides with that of Lemma~\ref{Xn0} (ii);
if i) is satisfied, then 
${\bf Z}_{n_1}^0, {\bf Z}_{n_1-1}^0 \in D_{\bf m}, ~
   {\bf Z}_{n_1}^1, {\bf Z}_{n_1-1}^1 \in D_{{\bf m}'}$, 
   where ${\bf m}'-{\bf m} \in \{0,-1\}^g$ and $m_1=m_1'$;
if ii) is satisfied, then
${\bf Z}_{n_1}^0, {\bf Z}_{n_1-1}^0 \in D_{\bf m},~
   {\bf Z}_{n_1-1}^1, {\bf Z}_{n_1-2}^1 \in D_{{\bf m}'}$,
   where ${\bf m}'-{\bf m} \in \{0,-1\}^g$ and $m_1'-m_1=-1$.
\end{lemma}

\begin{proof}
Step 1: We show the first claim on $(Q_n^0,W_n^0)_n \in \mathcal{T}_C$. 
Recall that we have assumed $C_{g-1}>2C_g=0$, where $C_{g-1}$ and $C_g$ are given as \eqref{C's}.
From $C_g=0$, there exists $n_1$ such that $Q_{n_1}^0=0$ or $W_{n_1-1}^0=0$.
(a) Suppose $Q_{n_1}^0=0$, then
$$\min[\min_{i \neq n_1}Q_i, \min_{j\neq n_1-1,n_1}W_j ]\geq C_{g-1} >0,$$
and therefore $Q_i>0$ for $i\neq n_1$ and $W_j>0$ for $j\neq n_1-1,n_1$. 
(b) Suppose $W_{n_1'}=0$, then
$$\min[\min_{j \neq n_1'}W_j , \min_{i \neq n_1',n_1'+1} Q_i ] \geq C_{g-1} 
>0,$$
and therefore $W_j>0$ for $j\neq n_1'$ and $Q_i>0$ for $i\neq n_1',n_1'+1$.

The case where $Q_{n_1}=W_{n_1'}=0$ may occur only if 
$n_1=n_1'$ or $n_1=n_1'+1$, Both cases belong to i) not to ii). 

\noindent
Step 2: We assume ${\bf Z}_0\in D_{\bf 0}$ without loss of generality.
In the following we define $D_{(n_1,\ast)}$ as
$$
   D_{(n_1,\ast)} = 
   \bigcup_{n' \in \Z^g \text{ s.t. } n'_1 = n_1} D_{{\bf n}'}.
$$   
From Lemma \ref{Xn0} (ii), 
there exists unique $n_1 \in \{1, \cdots ,g+1\}$ such that 
${\bf Z}_{n_1}^0,{\bf Z}_{n_1-1}^0\in D_{\bf m}$ with $m_1=n_1-1$. 
From Lemma~\ref{regionZ} we have ${\bf Z}_{n_1-2}^0, {\bf Z}_{n_1-2}^1 
\in D_{(m_1-1,*)}$. 
Further, we have ${\bf Z}_{n_1}^1 \in D_{(m_1,*)}$,
indeed if ${\bf Z}_{n_1}^1 \in D_{(m_1-1,*)}$,
then ${\bf Z}_{n_1}^1, {\bf Z}_{n_1-1}^1,{\bf Z}_{n_1-2}^1
\in D_{(m_1-1,*)}$, which is a contradiction.
Thus ${\bf Z}_{n_1-1}^1$ belongs to $D_{(m_1,*)}$ or 
$D_{(m_1-1,*)}$.
If ${\bf Z}_{n_1-1}^1 \in D_{(m_1,*)}$, then 
\begin{align*}
Q_{n_1}^0&
=(\Theta({\bf Z}_{n_1-1}^0)-\Theta({\bf Z}_{n_1}^0))
-(\Theta({\bf Z}_{n_1-1}^1)-\Theta({\bf Z}_{n_1}^1))\\
&=m_1 L-m_1 L=0, 
\end{align*}
and if
${\bf Z}_{n_1-1}^1 \in D_{(m_1-1,*)}$ then 
\begin{align*}
W_{n_1-1}^0&
=L+(\Theta({\bf Z}_{n_1}^0)-\Theta({\bf Z}_{n_1-1}^0))-
(\Theta({\bf Z}_{n_1-1}^1)-\Theta({\bf Z}_{n_1-2}^1))\\
&=L-m_1 L+(m_1-1)L=0. 
\end{align*}

\noindent
Step 3. Since both $n_1$ in Step 1 and 2 are unique, they coincide.
\end{proof}

\begin{proposition}\label{sigma_inj}
The map $\sigma_t|_{{\rm Im} \iota_t}$ is injective.
\end{proposition}

\begin{proof}
It is enough to show the claim for $t=0$.
From the equation \eqref{WQ-T}
with the quasi-periodicity \eqref{T-qperiod}, 
we have
\begin{align}
\label{QW_T1}
T_{n'+1}^0-T_{n'}^0&=T_{n+1}^0-T_{n}^0-(n'-n)L
 +\sum_{j=n+1}^{n'}(Q_j^0+W_j^0) \quad (n, n' \in \Z, ~n \leq n'),\\
\label{QW_T2}
T_{n}^1-T_{n-1}^1&=T_{n}^0-T_{n-1}^0+Q_n^0 \quad (n\in \Z).
\end{align}
Take $n_0$ as $X_{n_0+1}^0=0$ (Lemma~\ref{Xn0}) and 
$n_1$ as Lemma~\ref{QW0} with
$1< n_1-n_0\leq g+1$.
Without loss of generality,  we can assume ${\bf Z}_{n_0}^0\in D_{\bf 0}$.
From $X_{n_0+1}^0=0$ and \eqref{UD-ptau}, we have 
$T_{n_0}^{2}-T_{n_0}^1=T_{n_0}^1-T_{n_0}^0$
and therefore we can also assume 
${\bf Z}_{n_0}^1$, ${\bf Z}_{n_0}^2$ $\in$ $D_{\bf 0}$.
From Lemma \ref{QW0}, 
we have ${\bf Z}_{n_1}^0, {\bf Z}_{n_1-1}^0 \in D_{\bf m}$ with $m_1=n_1-n_0-1$.

From \eqref{QW_T1}, $T_{n_0}^0=0$ and
$T_{n_1}^0-T_{n_1-1}^0 = -L(n_1-1-n_0),$
we have
\begin{align*}
T_n^0=& -\frac12(n-n_0)(n-n_0-1)L-S_n,
\end{align*}
where
$$
S_n=\left\{ \begin{array}{ll} 
\sum_{k=1}^{n-n_0} \sum_{j=n_0+k}^{n_1-1} (Q_j^0+W_j^0)
& \quad (n_0 \leq n \leq n_1)\\
\sum_{j=n_0+1}^{n_1-1}(j-n_0)(Q_j^0+W_j^0)
-\sum_{j=n_1}^{n-1} (n-j)(Q_j^0+W_j^0) &\quad (n \geq n_1)
\end{array}\right..
$$
We also have $T_n^t$ for $n<n_0$ by the quasi-periodicity.
Further, from \eqref{QW_T2} and $T_{n_0}^0=T_{n_0}^1=0$, we have
$$T_n^1=T_n^0+\sum_{j=n_0+1}^n Q_j^0.$$
\end{proof}

%%%%%%%%%%%%%%%%%%%%%%%%%%%%%%%%%%%%%%%%%%%%%%%%
\subsection{Proof of Proposition~\ref{iota_inj}}
%%%%%%%%%%%%%%%%%%%%%%%%%%%%%%%%%%%%%%%%%%%%%%%%
\label{part2}

We take $n_0$ as $X_{n_0+1}^0=0$ (Lemma~\ref{Xn0}), and 
set ${\bf Z}_{n_0}^0 \in D_{{\bf 0}}$ without loss of generality.
Then ${\bf Z}_{n_0}^1$ also belongs to $D_{{\bf 0}}$.
Take also $n_1$ as Lemma~\ref{QW0} with
$1< n_1-n_0\leq g+1$.
From \eqref{region}, Lemmas~\ref{regionZ} and \ref{Xn0}, we have 

\begin{align*}
  &{\bf Z}_{n_0}^0 \in D_{{\bf 0}} \\
  &{\bf Z}_n^0 \in D_{{\bf m}_n} ~~(n=1,\cdots,g) \\
  &{\bf Z}_{n_0+g+1}^0 \in D_{\vec{g}},
\end{align*}
where
\begin{align*}%\label{msk}
  &
  {\bf m}_n = \left\{ \begin{array}{ll}
  \sum_{k=1}^{n} \Ibf_{s(k)} & (n_0+1 \leq n \leq n_1-1)\\
   \sum_{k=1}^{n-1} \Ibf_{s(k)} & (n_1 \leq n \leq n_0+g+1) \\
  \end{array}\right.
\end{align*}
for some permutation $s \in \mathfrak{S}_g$.
Note that $\sum_{k=1}^g \Ibf_k = 
\sum_{k=1}^g \Ibf_{\sigma(k)} = \vec{g}$.

Let $\delta$ be a large positive number.
Set $\tilde{{\bf Z}}_n^t$ as 
\begin{align}\label{Z-delta}
\tilde{{\bf Z}}_n^t={\bf Z}_n^t-\delta (n-n_0) {\bf e}_1
={\bf Z}_{n_0}^0 -(L+\delta)(n-n_0) {\bf e}_1+t\vec{\lambda}
\end{align}
and $\tilde{K}$ as
$$\tilde{K}=K+\delta J,\quad J=\left(\begin{array}{ccccc}
2&-1&&&\\
-1&2&-1&&\\
&\ddots&\ddots&\ddots&\\
&&-1&2&-1\\
&&&-1&2
\end{array}\right),$$
where the blank parts in $J$ means zero.
For ${\bf m} \in \R^g$, let $\tilde{D}_{{\bf m}}$ be the fundamental
regions determined by $\tilde{K}$.

\begin{lemma}\label{Ztilde}
If ${\bf Z}_n^t \in D_{{\bf m}}$, then $\tilde{{\bf Z}}_n^t$ belongs to 
$\tilde{D}_{{\bf m}}$ for all $n\in \Z$ and $t=0,1$.
\end{lemma}

\begin{proof}
From \eqref{region} it is enough to show 
\begin{align} \label{zt1}
-{\bf l}({\bf Z}_n^t -\delta (n-n_0) {\bf e}_1)^{\bot}
\leq {\bf l}(K+\delta J)({\bf m} + \frac12 {\bf l})^{\bot},
\end{align}
under the constraints
$$ -{\bf l} ({\bf Z}_n^t)^{\bot} 
\leq {\bf l}K({\bf m} + \frac12 {\bf l})^{\bot},$$
for any ${\bf l}=\Ibf_j - \Ibf_i $ ($0\leq i , j \leq g$),
$n_0+1 \leq n \leq n_0+g+1$ and $t=0,1$, where ${\bf m}$ is 
$\sum_{k=1}^{n-n_0} \Ibf_{s(k)}$ or $\sum_{k=1}^{n-n_0-1} \Ibf_{s(k)}$.
"R.h.s $-$ l.h.s" of \eqref{zt1} is greater than or equal to
\begin{align*}
\delta {\bf l}J({\bf m}+\frac12 {\bf l})^{\bot}-\delta(n-n_0)l_1,
\end{align*}
which can be shown to be nonnegative by using the formulae
\begin{align}\label{zt2}
J \Ibf_k&= \left\{\begin{array}{ll}
{\bf e}_1+{\bf e}_k- {\bf e}_{k+1} &(1\leq k \leq g-1) \\
 {\bf e}_1+{\bf e}_g &(k=g),
\end{array}\right.
\qquad \frac{1}{2}{\bf l}K{\bf l}^{\bot}=1, 
\end{align}
and
\begin{align*}
{\bf l}J{\bf m}^{\bot} &\geq \left\{\begin{array}{ll}
n-n_0-1 &(l_1=1) \\
 -1 &(l_1=0)\\
-n+n_0-1 &(l_1=-1)
\end{array}\right. .
\end{align*}
\end{proof}

\begin{lemma}\label{QWtilde}
Define $\tilde{Q}_n^t$ and $\tilde{W}_n^t$ as
\begin{align*}
\tilde{Q}_n^t&=
 \tilde{\Theta}(\tilde{{\bf Z}}_{n-1}^{t}) 
 + \tilde{\Theta}(\tilde{{\bf Z}}_{n}^{t+1}) 
 - \tilde{\Theta}(\tilde{{\bf Z}}_{n-1}^{t+1}) 
 - \tilde{\Theta}(\tilde{{\bf Z}}_{n}^{t}),\\
 \tilde{W}_n^t &= L+\delta + \tilde{\Theta}(\tilde{{\bf Z}}_{n-1}^{t+1}) 
 + \tilde{\Theta}(\tilde{{\bf Z}}_{n+1}^{t})
 - \tilde{\Theta}(\tilde{{\bf Z}}_{n}^{t}) 
 - \tilde{\Theta}(\tilde{{\bf Z}}_{n}^{t+1}),
\end{align*} 
where $\tilde{\Theta}(\tilde{Z}) = \Theta(\tilde{Z};\tilde{K})$.
Then $\tilde{Q}_n^0=Q_n^0$ and $\tilde{W}_n^0=W_n^0$ hold except
$\tilde{W}_{n_0}^0=W_{n_0}^0+\delta$.
\end{lemma}

\begin{proof}
From Lemma~\ref{Ztilde} 
we obtain the following formula;
$$\tilde{\Theta}(\tilde{{\bf Z}}_n^t)-\Theta({\bf Z}_n^t)
= \left\{\begin{array}{ll}
 -\frac{\delta}{2}(n-n_0)(n-n_0-1)
&(n_0\leq n \leq n_0+g+1)\\
0& (n=n_0-1)\end{array}\right.$$
for $t=0,1$. Indeed we have
\begin{align*}
&\tilde{\Theta}(\tilde{{\bf Z}}_n^t)-\Theta({\bf Z}_n^t)\\
=& \frac{1}{2}{\bf m} (K+\delta J) {\bf m}^\bot
+{\bf m}({\bf Z}_n^t-\delta(n-n_0){\bf e}_1)^\bot 
- \frac{1}{2}{\bf m}K {\bf m}^\bot
-{\bf m}{\bf Z}_n^{t \bot} \\
=& \frac{1}{2}\delta{\bf m}J{\bf m}^\bot -\delta (n-n_0)m_1\\
=& \left\{\begin{array}{ll}
\frac{\delta}{2}((n-n_0)^2+(n-n_0))-\delta (n-n_0)^2 & (m_1=n-n_0)\\
\frac{\delta}{2}((n-n_0-1)^2+(n-n_0-1))-\delta (n-n_0)(n-n_0-1) 
& (m_1=n-n_0-1)\\
\end{array}\right.\\
=& -\frac{\delta}{2}(n-n_0)(n-n_0-1),
\end{align*}
for $n_0\leq n \leq n_0+g+1$, and the case of $n=n_0-1$ is similarly shown.
Thus we have
\begin{align*}
\tilde{Q}_n^0=&
 \Theta({\bf Z}_{n-1}^{0}) 
 + \Theta({\bf Z}_{n}^{1}) 
 - \Theta({\bf Z}_{n-1}^{1}) 
 - \Theta({\bf Z}_{n}^{0})\\
&-\frac{\delta}{2}
((n-n_0-1)(n-n_0-2)+(n-n_0)(n-n_0-1)\\
&\ -(n-n_0-1)(n-n_0-2)-(n-n_0)(n-n_0-1))\\
=&Q_n^0
\end{align*}
for $n_0+1 \leq n \leq n_0+g+1$,
\begin{align*}
W_n^0 =& L+\delta + \Theta({\bf Z}_{n-1}^{1}) 
 + \Theta({\bf Z}_{n+1}^{0})
 - \Theta({\bf Z}_{n}^{0}) 
 - \Theta({\bf Z}_{n}^{1})\\
&-\frac{\delta}{2}
((n-n_0-1)(n-n_0-2)+(n-n_0+1)(n-n_0)\\
&\ -(n-n_0)(n-n_0-1)-(n-n_0)(n-n_0-1))\\
=&W_n^0
\end{align*}
for $n_0+1 \leq n \leq n_0+g$, and 
\begin{align*}
W_{n_0}^0 =& L+\delta + \Theta({\bf Z}_{n_0-1}^{1}) 
 + \Theta({\bf Z}_{n_0+1}^{0})
 - \Theta({\bf Z}_{n_0}^{0}) 
 - \Theta({\bf Z}_{n_0}^{1})\\
=&W_{n_0}^0+\delta.
\end{align*}
\end{proof}

\begin{lemma}\label{Dtilde}
Set ${\bf m}_n=\sum_{k=1}^{n-n_0-1}\Ibf_k$. 
For $1 \ll t \ll \delta$, we have 
$$
  \tilde{{\bf Z}}_{n}^t \in 
  \begin{cases} 
    \tilde{D}_{\bf 0} ~~ \text{for } n= n_0, \\
    \tilde{D}_{{\bf m}_{n-n_0-1}} ~~\text{for } n_0+1\leq n\leq n_0+g+1.
  \end{cases}
$$
\end{lemma}

\begin{proof}
From \eqref{region} it is enough to show that 
$$ -{\bf l}({\bf Z}_{n_0} -(n-n_0)(L+\delta){\bf e}_1+t\vec{\lambda})^{\bot}
\leq {\bf l}(K+\delta J)
\left(\sum_{k=1}^{n-n_0-1}\Ibf_k + \frac12 {\bf l}\right)^{\bot}$$
for any ${\bf l}=\Ibf_j - \Ibf_i $ ($0\leq i, j \leq g$).
Putting together all terms independent of $\delta$ and $t$ into $c$, 
this is equivalent to
\begin{align}
&c + 
t{\bf l}\vec{\lambda}^{\bot}
+\delta \left({\bf l}J
\left(\sum_{k=1}^{n-n_0-1}\Ibf_k + \frac12 {\bf l}\right)^{\bot}
-(n-n_0)l_1 \right) \geq 0. \label{dt1}
\end{align}
From formula \eqref{zt2} we have $J{\bf m}_n=(n-n_0){\bf e}_1-{\bf e}_{n-n_0}$,
and thus \eqref{dt1} becomes 
\begin{align}
c+ t{\bf l}\vec{\lambda}^{\bot} + \delta(-l_{n-n_0}+1)\geq 0. \label{dt2}
\end{align}
If $l_{n-n_0} = 0$ or $-1$, \eqref{dt2} holds obviously for 
$1 \ll t \ll \delta$,
while if $l_{n-n_0}=1$, \eqref{dt2} holds for large $t$ since
${\bf l}\vec{\lambda}^{\bot} >0$.
\end{proof}

\begin{proposition}\label{iota_inj}
The map $\iota_t$ is injective.
\end{proposition}

\begin{proof}
For $(Q_n^0,W_n^0)_n \in {\rm Im}\, \sigma_0 \circ \iota_0$ and 
a large positive number $\delta$, $(\tilde{Q}_n^0,\tilde{W}_n^0)_n$ 
is determined by using Lemma~\ref{QWtilde}, and thus $\tilde{T}_n^t$ are 
uniquely determined by Proposition~\ref{sigma_inj}.
Further, from Lemma~\ref{Dtilde}
we know the region $\tilde{D}_{{\bf m}_n}$ the point 
$\tilde{{\bf Z}}_n^t$ belongs to for $1 \ll t \ll \delta$, thus
we obtain $g$ independent linear equations for $\tilde{{\bf Z}}_0$,
which uniquely determine $\tilde{{\bf Z}}_0$ (as explained in the example). 
Finally we get a unique ${\bf Z}_0$ through
${\bf Z}_{n_0} = \tilde{{\bf Z}}_{n_0}$ and \eqref{Z-delta}.
\end{proof}

%%%%%%%%%%%%%%%%%%%%%%%%%%%%%%%%%%%%%%%%%%
\subsection{Proof of Theorem \ref{J-Toda}}
%%%%%%%%%%%%%%%%%%%%%%%%%%%%%%%%%%%%%%%%%%
\label{part3}
From Lemma~\ref{sigma_inj} and Lemma~\ref{iota_inj}, 
$\iota_\sigma$ is injective.
We show that $\iota_\sigma$ is surjective.
Since $\iota_\sigma$ is continuous both for ${\bf Z}_0$ and 
$C=(C_{-1},C_0,\dots,C_g)$, 
it is enough to show
that $\iota_\sigma$ is a surjective map from a dense subset of $J(K)$ 
to a dense subset of $\mathcal{T}_C$ for a dense subset of
$$\{ C \in \R^{g+2} ~|~ \mbox{$C$ satisfies the generic condition 
\eqref{CD-condition}} \}.$$ 
Suppose that $C$ belongs to $\Q^g$. Let $N$ be a common multiple of
denominators of $C_i$'s. Let $J(K)_N$ denote 
$J(K)\cap (\frac{1}{N}\Z)^g$ and 
$(\mathcal{T}_C)_N=(\mathcal{T}_C)\cap (\frac{1}{N}\Z)^g$. 
Here $J(K)_N$ and $(\mathcal{T}_C)_N$ are isomorphic to 
$J(NK)_1$ and $(\mathcal{T}_{NC})_1$ (the set of integer points
of the phase space with the conserved quantities 
$NC=N(C_{-1},\dots,C_g)$) respectively.
From Proposition~4.4 of 
\cite{InoueTakenawa08}, 
we have $|(\mathcal{T}_{NC})_1|=(g+1)|B_{NL,N\lambda}|$ 
where $B_{NL,N\lambda}$ denotes the isolevel set 
of the corresponding pBBS.
From Proposition~2.2 of \cite{YYT03} it is easily proved by using 
``'10' elimination'' (or Theorem~3.11 of \cite{KTT06}) that 
the cardinal number $|B_{NL,N\lambda}|$ of 
$B_{NL,N\lambda}$ is $N^g p_0p_1\cdots p_{g-1}$,
which is a $(g+1)$-th part of $|J(NK)_1|$
from Lemma~2.5 of \cite{InoueTakenawa08}.
Therefore, since $\iota_\sigma$ is injective, $\iota_\sigma$
is a surjection from $J(K)_N$ to $(\mathcal{T}_C)_N$. Since $N$
can be taken as large as one likes and $\Q^g$ is dense in $\R^g$,
$\iota_\sigma$ is a surjection from $J(K)$ to $\mathcal{T}_C$ for any 
$C$ satisfying the generic condition \eqref{CD-condition}.

%%%%%%%%%%%%%%%%%%%%%%%%%%%%%%%%
\subsection{Initial value problem}
%%%%%%%%%%%%%%%%%%%%%%%%%%%%%%%
\label{init-prob}
Using the arguments of this section, we can solve
the initial value problem for the UD-pToda.
Indeed, $T_n^0$ and $T_n^1$ can be determined from $(Q_n^0,W_n^0)_n$
as the proof of Proposition~\ref{sigma_inj}. 
Set $\delta$ be a large positive number as 
$\delta \simeq gL$ --- from the properties of the BBS 
(soliton like behavior \cite{TakahashiSastuma90}), 
the clusters of $1$'s 
align coherently in the order of their sizes for $t\simeq L$.
For $t\simeq L$, Lemma~\ref{Dtilde} gives
the region $\tilde{D}_{{\bf m}_n}$ to which the point 
$\tilde{{\bf Z}}_n^t$ belongs. As explained in the example
we obtain $g$ independent linear equations for $\tilde{{\bf Z}}_0$.
Solving these, we obtain $\tilde{{\bf Z}}_0$, and thus ${\bf Z}_{0}$ by
${\bf Z}_{n_0}^0 = \tilde{{\bf Z}}_{n_0}^0$.
Then the solution to the initial value problem
is given by \eqref{tau-theta} with \eqref{WQ-T}.

%%%%%%%%%%%%%%%%%%%%%%%%%%%%%
\subsection*{Acknowledgement}
%%%%%%%%%%%%%%%%%%%%%%%%%%%%%

T.~T. appreciates the assistance from the Japan Society for the
Promotion of Science.
R.~I. is supported by the Japan Society for the
Promotion of Science, Grand-in-Aid for Young Scientists (B) (19740231).

%%%%%%%%%%%%%%%%%%%%%%%%%%%%%%%%%

\end{document}